\documentclass{cccg16}
\usepackage{graphicx,amssymb,amsmath}

\usepackage{caption}

\DeclareMathOperator{\s}{S}

\newcommand{\R}{\mathcal{R}}

\title{Squarability of rectangle arrangements}

\author{Mat\v{e}j Kone\v{c}n\'{y}\thanks{Faculty of Mathematics and Physics,
Charles University in Prague, Czech Republic. E-mails: \texttt{matejkon@gmail.com, stanislav.kucera@outlook.com, opler@iuuk.mff.cuni.cz, j.sosnovec@email.cz, simsa.st@gmail.com, mtopfer@gmail.com}} \and Stanislav Ku\v{c}era\footnotemark[1] \and  Michal Opler\footnotemark[1] \and Jakub Sosnovec\footnotemark[1] \and \v{S}t\v{e}p\'{a}n \v{S}imsa\footnotemark[1] \and Martin T\"{o}pfer\footnotemark[1]}

\index{Kone\v{c}n\'{y}, Mat\v{e}j}
\index{Ku\v{c}era, Stanislav}
\index{Opler, Michal}
\index{Sosnovec, Jakub}
\index{\v{S}imsa, \v{S}t\v{e}p\'{a}n}
\index{T\"{o}pfer, Martin}

\begin{document}
\thispagestyle{empty}
\maketitle

\begin{abstract}
We study when an arrangement of axis-aligned rectangles can be transformed into an arrangement of axis-aligned squares in $\mathbb{R}^2$ while preserving its structure. We found a counterexample to the conjecture of J. Klawitter, M. N\"{o}llenburg and T. Ueckerdt whether all arrangements without crossing and side-piercing can be squared. Our counterexample also works in a more general case when we only need to preserve the intersection graph and we forbid side-piercing between squares. We also show counterexamples for transforming box arrangements into combinatorially equivalent hypercube arrangements. Finally, we introduce a linear program deciding whether an arrangement of rectangles can be squared in a more restrictive version where the order of all sides is preserved.
\end{abstract}

\section{Introduction}

In this paper, we are concerned with the following problem. Given an arrangement of axis-aligned rectangles in $\mathbb{R}^2$, is it possible to find an arrangement of axis-aligned squares with corresponding properties? J. Klawitter, M. N\"{o}llenburg and T. Ueckerdt \cite{main} asked which geometric rectangle arrangements can be transformed into combinatorially equivalent square arrangements. While showing some necessary and sufficient conditions for that, the question whether there exists an unsquarable rectangle arrangement without crossings and side-piercings (see Figure \ref{intersection_types}) remained open. We show a counterexample for that -- an arrangement of rectangles which is not combinatorially equivalent to any square arrangement. Moreover, our counterexample works even in a more general case when we only need to preserve the intersection graph of arrangements and we forbid side-piercing between squares.

In Section \ref{higherDimensions} we generalize the problem to higher dimensions -- considering hypercubes instead of squares and boxes instead of rectangles. We show that allowing crossings or side-piercings in any dimension leads to arrangements of boxes for which no corresponding arrangement of hypercubes exists.

Besides constructing counterexamples we also present an algorithm for deciding whether a given arrangement is squarable when the order of all sides has to be preserved (which implies combinatorial equivalence).

\subsection{Preliminaries}
Let $\R$ denote a given set of axis-aligned rectangles in $\mathbb{R}^2$ and $\s$ be a mapping from $\R$ to axis-aligned squares in $\mathbb{R}^2$ satisfying certain restrictions. If such $\s$ exists, we say that $\R$ is \emph{squarable} and $\s$ is a \emph{squaring} of $\R$. Thus $\s(\R)$ is a set of squares obtained from $\R$ in a way specific to the particular variant and $\s(R)$ is the square representing the rectangle $R\in\R$. In each variant we explain the restrictions placed on the input set of rectangles $\R$ and on the output set of squares $\s(\R)$.

\begin{figure}[h!]
\centering
\includegraphics[width=0.8\linewidth]{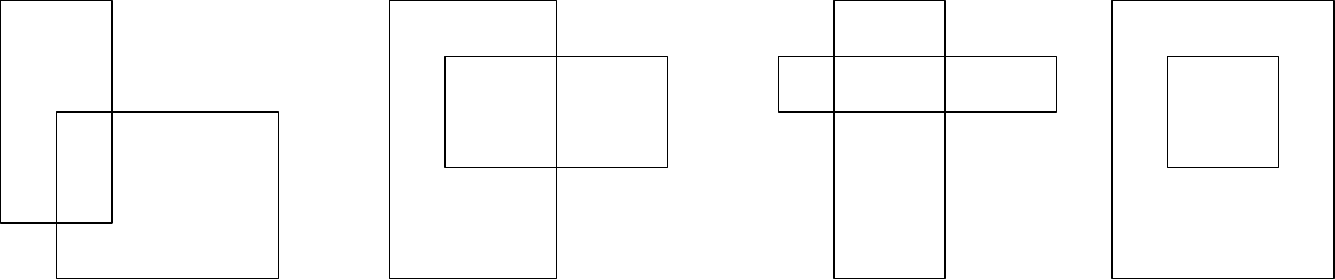}
\caption{Intersection types. Respectively: corner intersection, side-piercing, cross intersection and containment.}
\label{intersection_types}
\end{figure}

There are four intersection types: corner intersection, side-piercing, cross intersection and containment (see Figure \ref{intersection_types}). Note that we do not include empty intersection (formed by disjoint rectangles) as an intersection type. Also, we only consider sets of rectangles where no two rectangle sides are collinear.

In all the discussed variants, we assume that the input set $\R$ contains no two rectangles with side-piercing or cross intersection. Allowing these intersection types easily leads to instances of arrangements of rectangles that cannot be squared -- any two rectangles with the cross intersection clearly cannot be squared as well as the arrangement of four rectangles in Figure~\ref{piercing} for side-piercing.

\begin{figure}[h!]
\centering
\includegraphics[width=0.4\linewidth]{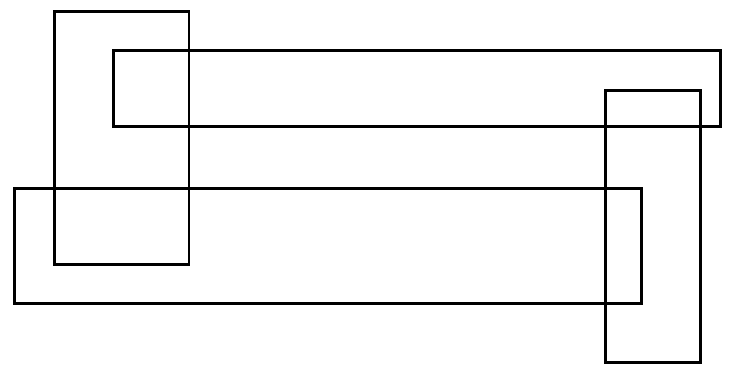}
\caption{An arrangement that cannot be squared due to side-piercing intersections.}
\label{piercing}
\end{figure}

Without loss of generality, we assume all the rectangles have positive coordinates. If it is not the case we just translate the whole arrangement. For a rectangle $R$ we denote:
\begin{itemize}
    \item $t(R)$ to be the $y$-coordinate of the top side of $R$,
    \item $b(R)$ to be the $y$-coordinate of the bottom side of $R$,
    \item $r(R)$ to be the $x$-coordinate of the right side of $R$,
    \item $l(R)$ to be the $x$-coordinate of the left side of $R$,
    \item $h(R)$ to be the height of $R$: $h(R)=t(R)-b(R)$,
    \item $w(R)$ to be the width of $R$: $w(R)=r(R)-l(R)$.
\end{itemize}

\subsection{Variants of the squarability problem}
Let $\R$ be an arrangement of rectangles and $\s$ be a squaring of $\R$. We say that $\R$ and $\s(\R)$ are \emph{combinatorially equivalent} if for any $R_1,R_2\in\R$, the intersection type of $\s(R_1)$ and $\s(R_2)$ is the same as the intersection type of $R_1$ and $R_2$ and these intersections happen exactly on the same sides (and corners). For example, if $R_1$ and $R_2$ have corner intersection that is in the upper left corner on $R_1$ and the lower right corner of $R_2$, the same must hold for $\s(R_1)$ and $\s(R_2)$.

Note that the above definition of combinatorial equivalence is strictly weaker than the one given in \cite{main}. This definition is, however, convenient to us as the basic requirement. Since our counterexample works in this less restrictive case, it is also a counterexample when the referenced definition is used.

The following are variants of the squarability problem. They vary in the strength of the assumptions we put on the mapping $\s$.
\begin{description}
\item[Preserve order of all sides.]
The output $\s(\R)$ has to be combinatorially equivalent to $\R$ and, moreover, the respective order of sides on both axes has to be preserved. On a chosen axis, we can construct the sequence of sides of rectangles $\R$ from left to right as they appear, i.e., every rectangle will appear exactly twice. Then the same sequence of sides has to be realized in $\s(\R)$.
\item[Combinatorial equivalence.]
The output $\s(\R)$ has to be combinatorially equivalent.
\item[Keep intersections, forbid side-piercing.] First, we require that the intersection graphs of $\R$ and $\s(\R)$ are isomorphic, i.e., it holds that $R_1\cap R_2\neq\emptyset$ if and only if $\s(R_1)\cap\s(R_2)\neq\emptyset$ for all $R_1,R_2\in\R$. Additionally, the squares in the output set $\s(\R)$ must only have corner intersections or containment.

\item[Keep intersection graph.] We only require that the intersection graphs of $\R$ and $\s(\R)$ are isomorphic. 
\end{description}

Note that if $\s$ satisfies ``Preserve order of all sides'', then it satisfies ``Combinatorial equivalence''. In the same sense, ``Combinatorial equivalence'' implies ``Keep intersection, forbid side-piercing'' (by the assumption that $\R$ contains no side-piercing), which implies ``Keep intersection graph''.


\section{Counterexamples}
In this section we will discuss examples of arrangements of rectangles, which cannot be squared in terms of the mapping $\s$. In each subsection we consider squarability with respect to of one of the variants. We will start with the most restrictive case and proceed to more general variants.
\subsection{Preserving order of all sides}
If we want the resulting arrangement of squares to preserve the order of all sides, there is an easy example of four rectangles that cannot be squared.\\
\begin{figure}[h!]
\centering
\includegraphics[width=25mm]{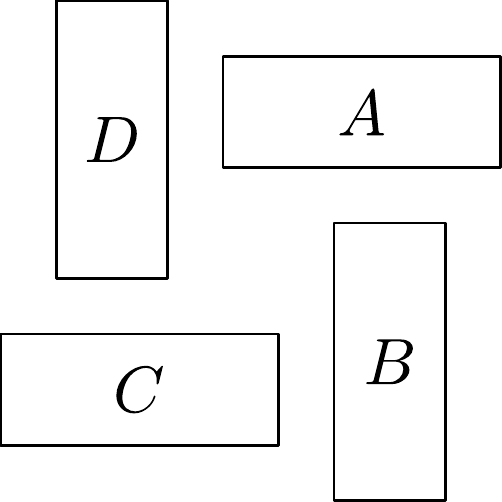}
\caption{An arrangement not squarable in the most restrictive case.}
\label{fig:firstcounter}
\end{figure}

\begin{theorem}\label{4rect}
The arrangement of rectangles in Figure \ref{fig:firstcounter} cannot be squared while preserving order of all sides.
\end{theorem}

\begin{proof}
After squaring the arrangement we would get $w(A)>w(B)=h(B)>h(C)=w(C)>w(D)=h(D)>h(A)=w(A)$; thus, the arrangement is unsquarable.
\end{proof}

This is an easy observation but it is important, because this arrangement is exactly the one we will find in latter cases to prove unsquarability of other arrangements.
\subsection{Combinatorial equivalence}
In the second most restrictive definition of the mapping $\s$ we want the resulting arrangement of squares to not only have the same types of intersections but also to have the same position. This means that if there is a rectangle A and a rectangle B intersecting A in the top right corner then $\s(B)$ will intersect $\s(A)$ again in the top right corner.

\begin{figure}[h!]
\centering\includegraphics[width=0.45\textwidth]{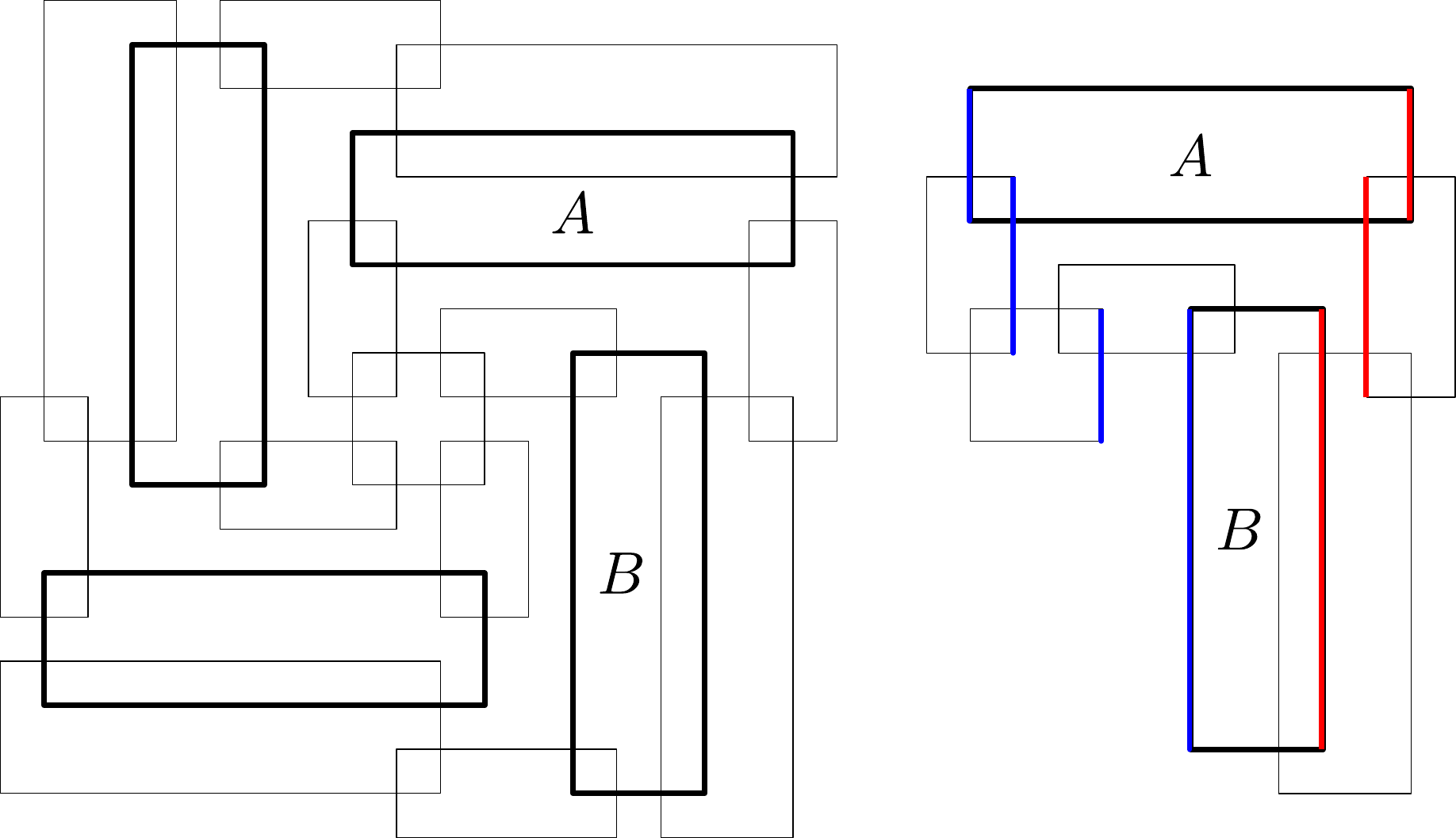}
\caption{An arrangement not squarable when $\s$ keeps the combinatorial equivalence.}
\label{fig:secondcounter}
\end{figure}

\begin{theorem}
The arrangement of rectangles in the left picture of Figure \ref{fig:secondcounter} cannot be squared.
\end{theorem}
\begin{proof}
To prove that, we want to show that the four bold rectangles form the pattern from Theorem \ref{4rect}. To do that, we need to prove that there is that cyclic condition on lengths of their sides. It suffices to show the dependency only for one pair of neighbouring rectangles since the arrangement is symmetric.

In the right picture of Figure \ref{fig:secondcounter}, there is the situation for $A$ and $B$ where only the important rectangles are drawn. Suppose the rectangles are orientated as in the picture (orientation is fixed for the whole arrangement). To prove $w(A)>w(B)$ in all possible mappings $\s$, it is sufficient to show $l(A)<l(B)$ and $r(A)>r(B)$.

We observe that when two rectangles $C$ and $D$ intersect a~common rectangle $E$ on its top (or bottom) side, $C$ being the one intersecting it in the left corner, and $C,D$ do not intersect each other it must hold $r(C)<l(D)$. When two rectangles $F,G$ intersect each other then $l(F)<r(G)$. These two observations used on the red sides of the rectangles in Figure \ref{fig:secondcounter} together give us $r(A)>r(B)$. To prove $l(A)<l(B)$ we use the observations for the blue sides.
\end{proof}

\subsection{Keep intersections, forbid side-piercing}
So far we have been mainly building tools and considering easy examples. For $\s$ which only keeps intersections without allowing side-piercing in $\s(\R)$ we still need one more tool.

\begin{figure}[h!]
\centering\includegraphics[width=0.9\linewidth]{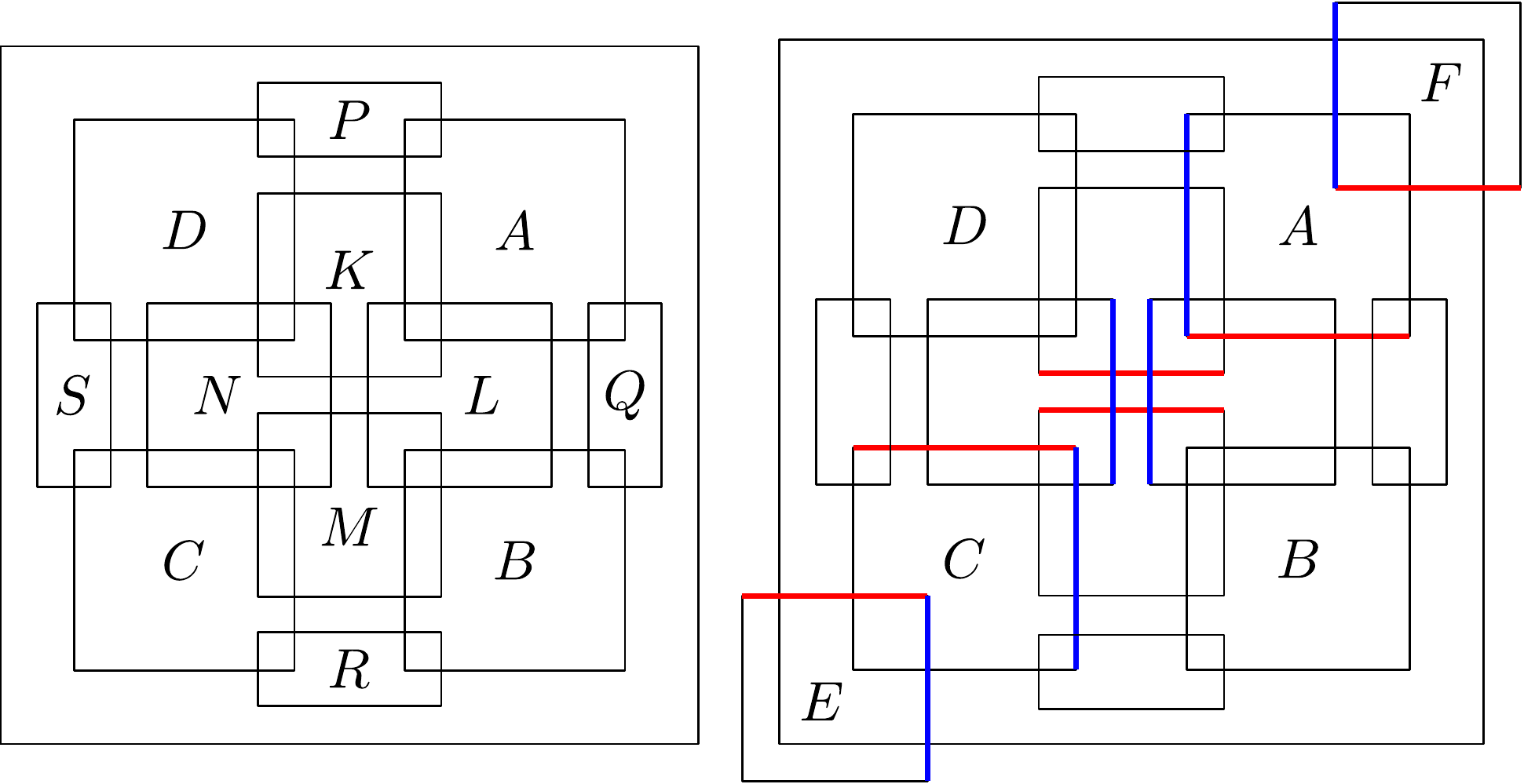}
\caption{$\Sigma$-gadget and its usage.}
\label{fig:structure}
\end{figure}

We refer to the arrangement depicted in the left picture of Figure \ref{fig:structure} as a \emph{$\Sigma$-gadget}. It is an arrangement of rectangles that can be squared even in the most restrictive case and we use it to force some useful properties.

\begin{lemma}
All squarings of the $\Sigma$-gadget that keep intersections but forbid side-piercings are combinatorially equivalent, up to rotation and reflection.
\end{lemma}
\begin{proof}
First look at the rectangles $K, L, M, N$ in the middle. There is only one way to square them upon a rotation and reflection. Then we want to square rectangles $A,B,C$ and $D$. Notice that $A$ can be contained neither in $K$ nor in $L$ because it intersects $P$. This, and the fact the side-piercing is forbidden, gives us three possibilities how to place $A$, relatively to $K$ and $L$. It can be either in the position as in the Figure \ref{fig:structure} or in such position that it contains the intersection of the squares $K$ and $L$ or in the opposite corner than in the first case. In Figure~\ref{fig:3options} we see all the important cases.
\begin{figure}[h!]
\centering
\includegraphics[width=0.8\linewidth]{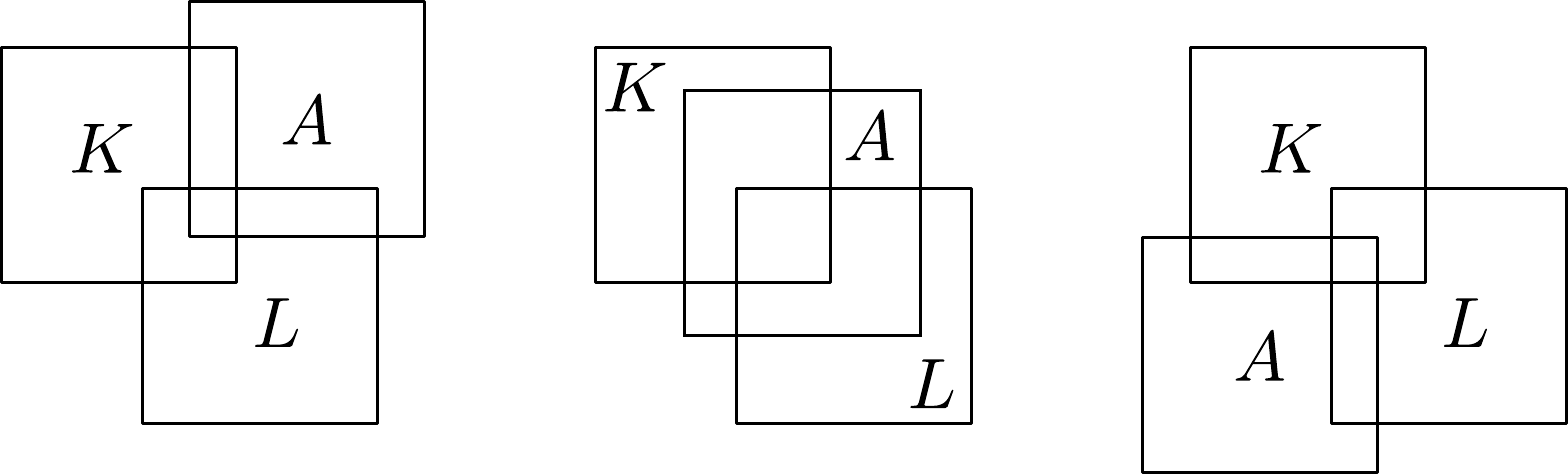}
\caption{Three possible ways of placing rectangle $A$.}
\label{fig:3options}
\end{figure}
The first case is the one we want. In the second case, the position of $A$ forces $P$ (and $Q$) to intersect the bottom left corner of $A$ because $P$ ($Q$) needs to intersect $D$ ($B$) without intersecting $K$ ($L$). This means $P$ and $Q$ both intersect the bottom left corner of $A$ and so they intersect each other, a contradiction. In the last case, $A$ would intersect $M$ and $N$, a contradiction. Therefore, there is only one way to square $A$ and by symmetry the same is true for $B, C$ and $D$. Now the rectangles $P, Q, R$ and $S$ can also be squared in only one way, completing the proof.
\end{proof}

First we explain how we use the $\Sigma$-gadget in an arrangement. If we want another rectangle (or another $\Sigma$-gadget) to intersect our $\Sigma$-gadget in a corner, it must intersect both the surrounding rectangle and one of $A,B,C$ or $D$ depending on in which corner it intersects $\Sigma$-gadget. Besides these two it does not intersect anything else.

Now that we know that the $\Sigma$-gadget can be squared in exactly one way and how to use it in an arrangement, let us explore some of its useful properties. As is illustrated in the right picture of Figure \ref{fig:structure}, the most useful property comes to play when the $\Sigma$-gadget is intersected by rectangles in opposite corners, lets call them $E$ and $F$. Usually this only gives us one of the following conditions:
\begin{itemize}
    \item $r(E)<l(F)$ (blue colored sides in the picture).
    \item $t(E)<b(F)$ (red colored sides in the picture).
\end{itemize}
The $\Sigma$-gadget provides both conditions at the same time, which is very useful when forcing the situation like in Theorem \ref{4rect}. At the same time, if the $\Sigma$-gadget is intersected in two corners, we can always say whether the corners are \emph{opposite} or \emph{adjacent}. For the purposes of arrangements in which the $\Sigma$-gadget is used, when we talk about the height, width, left side and so on, we always mean the height, width, left side, ... of the outer rectangle.

\begin{figure}[h!]
\centering
\includegraphics[width=0.9\linewidth]{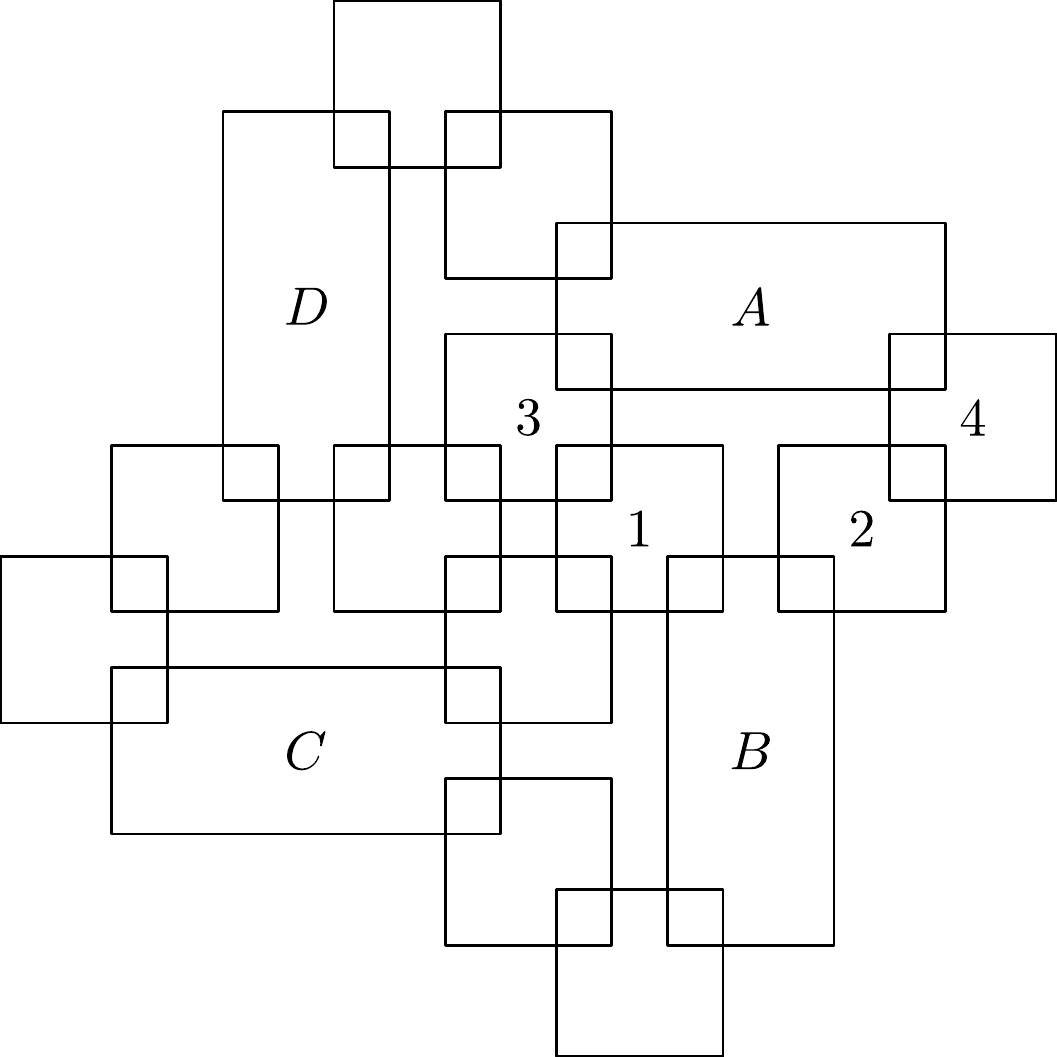}
\caption{An arrangement using $\Sigma$-gadget not squarable even in the least restrictive case without side-piercing.}
\label{fig:finalcounter}
\end{figure}
Having such a strong tool it is now easy to create an arrangement of rectangles that cannot be squared.

\begin{theorem}
The arrangement from Figure \ref{fig:finalcounter} with
the $\Sigma$-gadget instead of each rectangle cannot be squared.
\end{theorem}
\begin{proof}
We show that rectangles $A$, $B$, $C$ and $D$ form the same arrangement as we saw in Theorem \ref{4rect}. Rectangles 1 and 2 lie on the same side of $B$. Rectangles 3 and 4 lie in the opposite corners of 1 and 2 respectively with respect to $B$. Because rectangles 1 and 2 are $\Sigma$-gadgets, this implies $l(B)>r(3)$ and $r(3)>l(A)$ since rectangles $A$ and 3 intersect each other. We showed $l(B)>l(A)$ and similarly using rectangles 2 and 4 we can show $r(B)<r(A)$. Together this gives us $w(B)<w(A)$. Rotating the argument around the arrangement we show that if the arrangement gets squared it holds $w(B)<w(A)=h(A)<h(D)=w(D)<w(C)=h(C)<h(B)=w(B)$, which cannot be true.
\end{proof}

One could think this cannot be all. After all in previous cases we needed to show there is only one way to draw the arrangement of squares and we always ended up with a square which we couldn't add. Note that we did just that by showing the $\Sigma$-gadget can be squared in only one way.

\section{Higher dimensions}
\label{higherDimensions}
In this section we will make some observations about arrangements of boxes in higher dimensions. We use the same notation as before, that is $\R$ denotes a set of axis-aligned boxes in $\mathbb{R}^d$ and $S$ its mapping to a set of axis-aligned hypercubes in $\mathbb{R}^d$. We will often work with projections of $\mathbb{R}^d$ to a subset of coordinates. For a set $I \subseteq \{1, \ldots, d\}$ let $\mu_I: \mathbb{R}^d \to \mathbb{R}^{|I|}$ be a projection that ``forgets'' all coordinates not indexed by $I$. Furthermore for a singleton-indexed projection we shorten its notation $\mu_{\{c\}} = \mu_c$. The result of a projection $\mu_I$ applied to $\R$ is an arrangement of axis-aligned boxes or hypercubes in $\mathbb{R}^{|I|}$.

\begin{figure}[h!]
\centering
\includegraphics[width=0.9\linewidth]{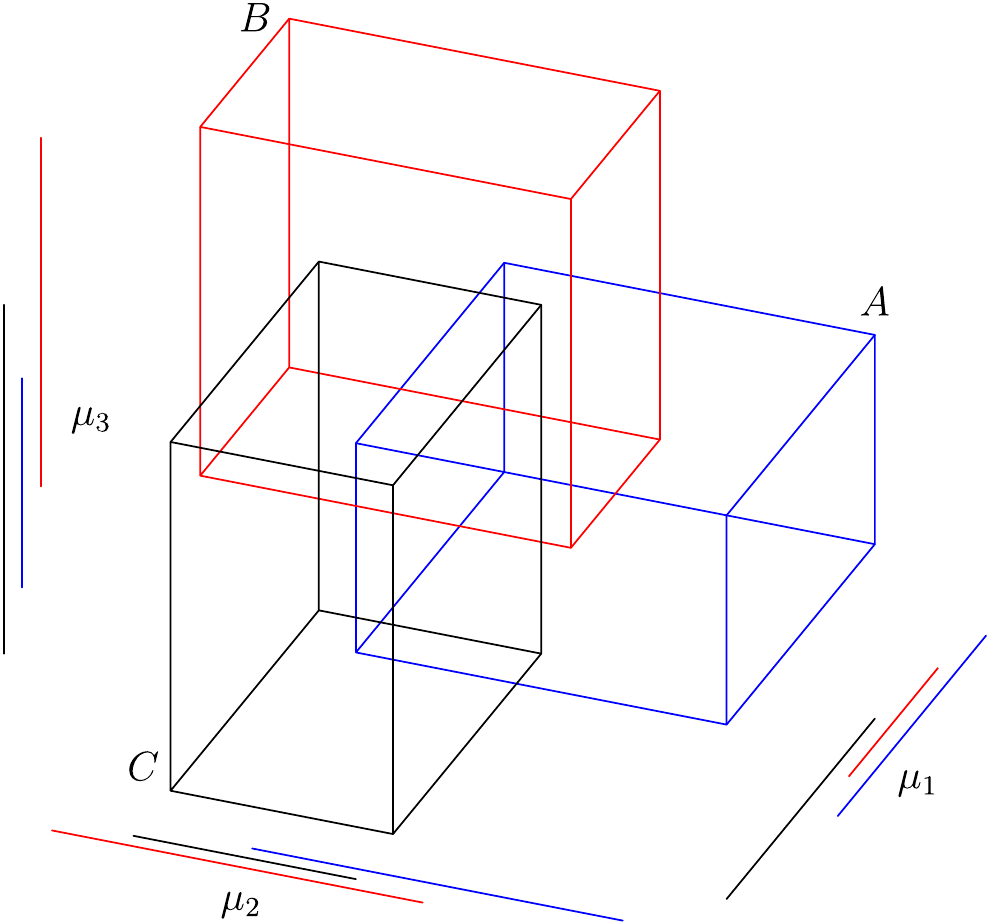}
\caption{An arrangement of three boxes in $\mathbb{R}^3$ such that there is no combinatorially equivalent arrangement of cubes.}
\label{fig:fig3d}
\end{figure}

The notion of combinatorial equivalence extends naturally to higher dimensions. We can observe that with each extra dimension we get new intersection types. For example, consider the following arrangement of only three boxes in $\mathbb{R}^3$ from Figure \ref{fig:fig3d}. Each pair of these boxes intersects in such a way that one pierces the edge of the other.  We claim that there cannot be a combinatorially equivalent arrangement of hypercubes.
Assume that we have such an arrangement of hypercubes $A'$, $B'$ and $C'$. Then the projection $\mu_1$ forces $A'$ to be bigger than $B'$, similarly $\mu_2$ for $B'$ and $C'$. Finally, $\mu_3$ forces $C'$ to be bigger than $A'$ and that is a contradiction.

\subsection{Boxicity and cubicity}
In the beginning we restricted to arrangements without side-piercings and cross intersections. It is fairly easy to see how they lead to counterexamples in the more restricted settings. However it is not so clear whether this restriction is needed in the least restrictive setting, i.e. preserving just the intersection graph. We will construct arrangements with these intersections which cannot be represented by an intersection graph of axis-aligned hypercubes up to a given dimension.

Let $G$ be a simple undirected graph. The \emph{boxicity} of $G$ is the smallest dimension $d$ such that $G$ can be represented as an intersection graph of axis-aligned boxes in $\mathbb{R}^d$. Similar notions are the \emph{cubicity} of $G$, where we consider a representation as an intersection graph of axis-aligned hypercubes, and the \emph{unit cubicity} of $G$, where all the hypercubes have to be unit. The notion of boxicity and unit cubicity (usually referred simply as cubicity) was introduced in 1969 by Roberts \cite{boxicity} and has since been actively studied, e.g. in~\cite{chandran2009upper}. The results we prove in this section were shown previously for unit cubicity in \cite{boxicity}. But our definition of cubicity is more general.

Furthermore, let $R(k, d)$ denote the smallest integer such that every coloring of the complete graph on $R(k,d)$ vertices with $d$ colors contains a monochromatic clique of size $k$. This is indeed one of the Ramsey numbers and it is well known that such a value exists.

As before, we want to construct such a graph that if there was an intersection-pattern equivalent arrangement of hypercubes it would force a cyclical inequality of hypercube sizes. However we do not have any tool yet for showing such inequalities in the most general setting.

\begin{lemma}
\label{lemma:Ramsey}
Let $G$ be a graph and $v$ be a vertex which has at least $R(k+2,d)$ neighbours that are pairwise non-adjacent. Suppose $G$ can be represented as an intersection graph of an arrangement $\R$ of axis-aligned hypercubes in $\mathbb{R}^d$ and $f: V(G) \to \R$ is the corresponding mapping. Then there is a neighbour $w$ of $v$ such that the hypercube $f(w)$ is more than $k$ times smaller than the hypercube $f(v)$.
\end{lemma}

\begin{proof}
Unsurprisingly, we will prove our claim using a coloring of the complete graph on $R(k+2, d)$ vertices. Each vertex gets labelled by one of the $R(k+2, d)$ neighbours of $v$. Observe that if two axis-aligned hypercubes $R_1$ and $R_2$ in $\mathbb{R}^d$ are disjoint then there is an integer $c$ such that $\mu_c(R_1)$ and $\mu_c(R_2)$ are disjoint. We will color an edge with any $c$ such that the corresponding hypercubes are disjoint under $\mu_c$. The number of vertices guarantees us a monochromatic clique of size $k+2$. That means there are $k+2$ neighbours of $f(v)$ that are pairwise disjoint under $\mu_c$ for some $c$. We have $k+2$ pairwise disjoint intervals and all of them need to intersect the interval $\mu_c(f(v))$. From this follows that the smallest interval $\mu_c(f(w))$ is more than $k$ times smaller than the interval $\mu_c(f(v))$. And since we are dealing with axis-aligned hypercubes, the hypercube $f(w)$ is more than $k$ times smaller than the hypercube $f(v)$. 
\end{proof}

\begin{theorem}
For every $d$ there is a graph $G$ with boxicity $2$ and cubicity larger than $d$.
\end{theorem}

\begin{proof}
Consider a complete bipartite graph $G$ with each partition of size $R(3,d)$. The boxicity of $G$ is 2 since one partition of $G$ can be represented as a set of vertical rectangles and the other as a set of horizontal rectangles (see Figure~\ref{fig:boxicity2}). Now suppose for a contradiction that the cubicity of $G$ is at most $d$ and fix any intersection representation with hypercubes in $\mathbb{R}^{d'}$, $d' \le d$. Let $v$ be the vertex of $G$ such that the corresponding hypercube is the smallest one. Since $v$ has exactly $R(3, d)$ pairwise disjoint neighbours, by Lemma \ref{lemma:Ramsey} there must be a neighbour of $v$ such that its corresponding hypercube is strictly smaller, that is a contradiction.
\end{proof}

\begin{figure}[h!]
\centering
\includegraphics[width=0.7\linewidth]{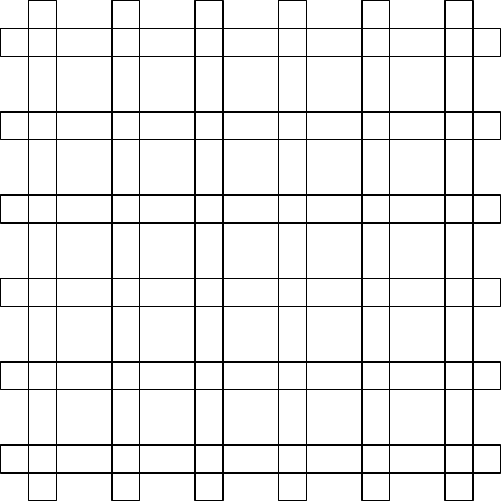}
\caption{An arrangement of rectangles whose intersection graph is a complete bipartite graph with each partition of size $R(3,2) = 6$.}
\label{fig:boxicity2}
\end{figure}


\section{Deciding squarability via LP}

\subsection{The problem}
In this section we present a linear program deciding whether a given arrangement of $n$ rectangles $\mathcal{R}=\{R_1,\ldots,R_n\}$ in $\mathbb{R}^2$ can be squared while preserving order of all sides.

Without loss of generality we can assume that all the endpoints of the intervals $[l(R_i),r(R_i)]$ and $[b(R_i),t(R_i)]$ have distinct values for all $i\in\{1,\ldots,n\}$. Otherwise we could change the endpoints a little without changing intersections between rectangles.

By ordering the endpoints of the intervals of projected rectangles into an increasing sequence, we obtain the sequence $a'_1< a'_2< \cdots< a'_{2n}$, where $a'_j=l(R_i)$ or $r(R_i)$ for some  $i\in\{1,\ldots,n\}$ (see Figure \ref{Fig1}). Replacing $l(R_i)$ and $r(R_i)$ by $i$ then yields the sequence $a_1,a_2,\ldots,a_{2n}$ of numbers $\{1,\ldots,k\}$, we call this sequence the $x$-sequence of $\mathcal{R}$. Clearly, each $i\in\{1,\ldots,n\}$ appears there exactly twice, thus for every $i\in\{1,\ldots,n\}$ we can define $a(i)=(j_1,j_2)$ such that $j_1<j_2$ and $a_{j_1}=a_{j_2}=i$. The $x$-sequence describes the respective ordering of the rectangles' $x$-coordinates. A $y$-sequence and the corresponding function $b$ are defined analogously.

\begin{figure}[h!]

  \centering
      \includegraphics[width=1\linewidth]{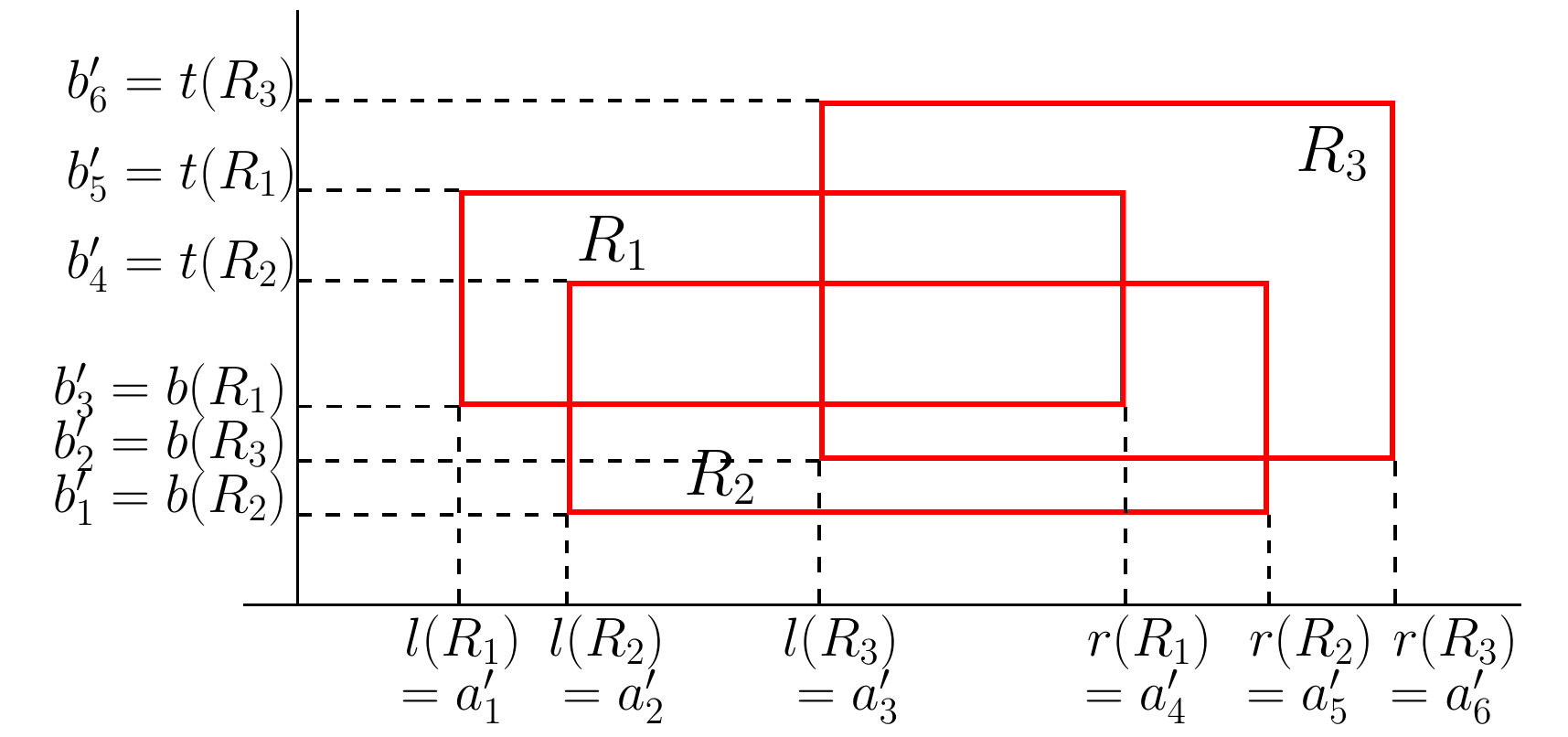}
  \caption{The $x$-sequence is $1,2,3,1,2,3$ and the $y$-sequence is $2,3,1,2,1,3$.}
 \label{Fig1}

\end{figure}

The decision problem can be reformulated in the following way: Given a family of rectangles $\mathcal{R}=\{R_1,\ldots,R_n\}$, does there exist a family of squares $\mathcal{S}=\{S_1,\ldots,S_n\}$ such that the $x$-sequence of $\mathcal{S}$ is identical to that of $\mathcal{R}$ and the $y$-sequence of $\mathcal{S}$ is identical to that of $\mathcal{R}$?

\subsection{Linear program}
Let us present a linear program solving the problem for an input set $\mathcal{R}=\{R_1,\ldots,R_n\}$. Let $a_1,a_1,\ldots,a_{2n}$ be the $x$-sequence and $b_1,b_2,\ldots,b_{2n}$ the $y$-sequence of $\mathcal{R}$. We have variables 
\[x_1,\ldots,x_{2n-1},y_1,\ldots,y_{2n-1}\geq 1,\]
where the value of $x_j$ represents the distance of the corresponding interval endpoints of rectangles $R_{a_{j}}$ and $R_{a_{j+1}}$ and the value of $y_j$ represents the distance of the corresponding endpoints of $R_{b_j}$ and $R_{b_{j+1}}$ (see Figure \ref{Fig2}). Let $(\mathbf{x},\mathbf{y})=(x_1,\ldots,x_{2n-1},y_1,\ldots,y_{2n-1})$ be any feasible solution to the following set of equalities. For every $i=1,\ldots,n$ we have an equality
\[\sum_{k=j_1}^{j_2-1}x_k=\sum_{k=j_1'}^{j_2'-1}y_k,\:\text{where }a(i)=(j_1,j_2),\,b(i)=(j_1',j_2')\]

\begin{figure}[h!]

  \centering
      \includegraphics[width=1\linewidth]{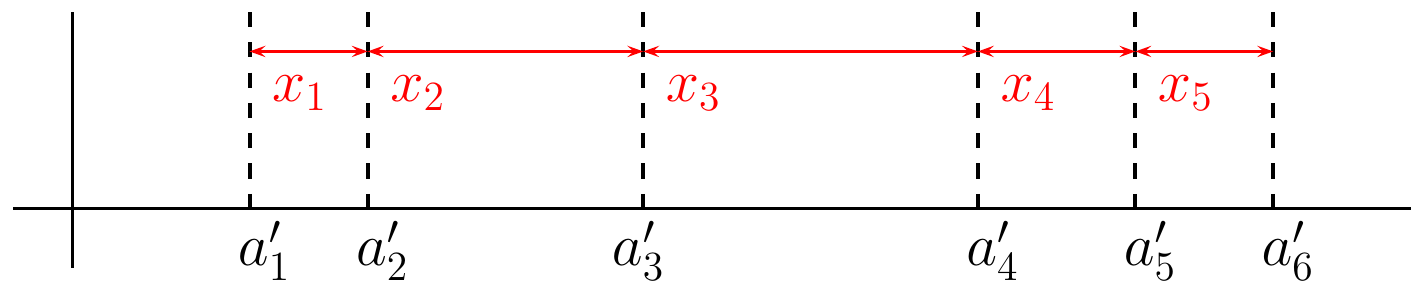}
  \caption{The meaning of the variables $x_1,\ldots,x_{2n-1}$.}
 \label{Fig2}

\end{figure}

From the solution $(\mathbf{x},\mathbf{y})$ we construct the corresponding set of squares $\mathcal{S}=\{S_1,\ldots,S_n\}$ as follows. Let $a(i)=(j_1,j_2)$ and $b(i)=(j_1',j_2')$, we set the coordinates of $S_i$ such that 
\[l(S_i)=\sum_{k=1}^{j_1-1}x_k,\quad r(S_i)=\sum_{k=1}^{j_2-1}x_k,\]
\[b(S_i)=\sum_{k=1}^{j_1'-1}y_k,\quad 
t(S_i)=\sum_{k=1}^{j_2'-1}y_k\]
As $x_i,y_i\geq 1$ for all $i\in\{1,\ldots,2n-1\}$, it is clear that the $x$-sequence and $y$-sequences of $\mathcal{R}$ are preserved in $\mathcal{S}$. The claim that $\mathcal{S}$ consists of squares follows immediately from the constraints of the linear program. Thus we obtain that if the linear program finds a feasible solution, we can construct an appropriate set of squares.

Reversely, let $\mathcal{S}$ be a set of squares that has the same $x$-sequence and $y$-sequence as $\mathcal{R}$. We can construct the variables $x_1,\ldots,x_{2n-1}$ and $y_1,\ldots,y_{2n-1}$ as the corresponding distances. It remains to sufficiently ``blow up'' this solution so that all of the variables are at least $1$. This is easily accomplished by multiplying the variables by the inverse of the minimum of them. We obtain a feasible solution to the linear program, as desired.

\section{Acknowledgments}
The authors would like to thank Pavel Valtr, Jan Kratochv\'{i}l and Stephen Kobourov for supervising the seminar where this paper was created. Our gratitude also goes to the anonymous referees for their helpful comments.

\small
\bibliographystyle{abbrv}
\bibliography{ref}

\end{document}